\documentclass[11pt]{article}

\usepackage{amsthm, amsbsy, amssymb, amsfonts}
\usepackage{fullpage,times}

\newcommand{\junk}[1]{}

\newtheorem{lemma}{Lemma}
\newtheorem{theorem}[lemma]{Theorem}

\newtheorem{definition}[lemma]{Definition}

\newtheorem{claim}{Claim}

\newcommand{\firstterm}{\mathsf{FirstTerm}}
\newcommand{\secondterm}{\mathsf{SecondTerm}}
\newcommand{\High}{\mathsf{High}}
\newcommand{\Medium}{\mathsf{Medium}}
\newcommand{\Low}{\mathsf{Low}}
\newcommand{\var}{\mathsf{Var}}
\newcommand{\ent}{\mathsf{H}}
\newcommand{\E}{\mathop{\mathbb E}}

\begin{document}

\title{On Zarankiewicz Problem and Depth-Two Superconcentrators\footnote{A preliminary 
version of this work appeared in the Proceedings of the
23rd International Symposium on Algorithms and Computation, 2012, pp. 257-266.}}

\author{Chinmoy Dutta\thanks{Twitter Inc., San Francisco, USA. 
  email: {\tt chinmoy@twitter.com}.} \\
  \and
  Jaikumar Radhakrishnan\thanks{Tata Institute of Fundamental Research, Mumbai, INDIA. 
  email: {\tt jaikumar@tifr.res.in}.}
}

\date{}
\maketitle

\begin{abstract}
We show tight necessary and sufficient conditions on the sizes of
small bipartite graphs whose union is a larger bipartite graph that
has no large bipartite independent set. Our main result is a common
generalization of two classical results in graph theory: the theorem
of K\H{o}v\'{a}ri, S\'{o}s and Tur\'{a}n on the minimum number of edges in
a bipartite graph that has no large independent set, and the theorem
of Hansel (also Katona and Szemer\'{e}di, Krichevskii) on the sum
of the sizes of bipartite graphs that can be used to construct a graph
(non-necessarily bipartite) that has no large independent set. 

As an application of our results, we show how they unify the underlying 
combinatorial principles developed in the proof of tight lower bounds for 
depth-two superconcentrators.
\end{abstract}

\section{Introduction}
\label{sec:intro}

Consider a bipartite graph $G=(V,W,E)$, where $|V|,|W|=n$. Suppose
every $k$ element subset $S \subseteq V$ is connected to every $k$
element subset $T \subseteq W$ by at least one edge. How many edges
must such a graph have? This is the celebrated Zarankiewicz
problem. 

\begin{definition}[Bipartite independent set]
A bipartite independent set of size $k \times k$ in a bipartite graph
$G=(V,W,E)$ is a pair of subsets $S \subseteq V$ and $T \subseteq W$
of size $k$ each such that there is no edge connecting $S$ and $T$, i.e.,
$(S \times T) \cap E = \emptyset$.
\end{definition}

\noindent The Zarankiewicz problem asks for the minimum number of edges in
a bipartite graph that does not have any bipartite independent set of
size $k \times k$. We may think of an edge as a complete bipartite
graph where each side of the bipartition is just a singleton. This
motivates the following generalization where we consider bipartite
graphs as formed by putting together not just edges, but, more
generally, small complete bipartite graphs.

\begin{definition}
A bipartite graph $G=(V,W,E)$ is said to be a union of complete
bipartite graphs $G_i=(V_i,W_i,E_i=V_i \times W_i)$, $i=1,2,\ldots,
r$, if each $V_i \subseteq V$, each $W_i \subseteq W$, and $E = E_1
\cup \cdots \cup E_r$.
\end{definition}

\begin{definition}
We say that a sequence of positive integers $(n_1,n_2,\ldots,n_r)$ is
$(n,k)$-strong if there is a bipartite graph $G=(V,W,E)$ that is a
union of graphs $G_i=(V_i,W_i, E_i=V_i \times W_i)$, $i=1,2,\ldots,r$,
such that
\begin{itemize}
\item $|V|, |W|=n$;
\item $|V_i|=|W_i|=n_i$;
\item $G$ has no bipartite independent set of size $k \times k$.
\end{itemize}
\end{definition}

\noindent What conditions must the $n_i$'s satisfy for
$(n_1,n_2,\ldots,n_r)$ to be $(n,k)$-strong?  Note that the
Zarankiewicz problem is a special case of this question where each
$n_i$ is $1$ and $\sum_i {n_i}$ corresponds to the number edges in the
final graph $G$.

\paragraph*{\bf Remark.} The Zarankiewicz 
problem is more commonly posed in the following form: What is the
maximum number of edges in a bipartite graph with no $k \times k$
bipartite {\em clique}. By interchanging edges and non-edges, we can
ask for the maximum number of non-edges (equivalently the minimum
number of edges) such that there is no $k \times k$ bipartite {\em
independent set}. This complementary form is more convenient for our
purposes.

\paragraph*{\bf Remark.} Given $n_i$'s, the choice of $G_i$'s also
determines whether or not $G$ has an independent set of a given
size. If $(n_1,n_2,\ldots,n_r)$ is $(n,k)$-strong, it implies that
there is {\em some} choice of $G_i$'s for which the union graph $G$
does not have any independent set of size $k \times k$.

\paragraph*{\bf The K\H{o}v\'{a}ri, S\'{o}s and Tur\'{a}n bound \\}

The following classical theorem gives a lower bound on the number of
edges in a bipartite graph that has no large independent set.

\begin{theorem}[K\H{o}v\'{a}ri, S\'{o}s and Tur\'{a}n~\cite{KST};
see, e.g., \cite{bollobas}, Page 301, Lemma 2.1.] 
\label{thm:kst}
If $G$ does not have an independent set of size $k\times k$, then
\[    n  {{n-{\overline d}} \choose k} {n \choose k}^{-1} \leq k-1, \]
where $\overline{d}$ is the average degree of $G$.
\end{theorem}
\noindent The above theorem implies that 
\begin{eqnarray*} 
n &\leq& (k-1) {{n-{\overline d}} \choose k}^{-1} {n \choose k}
  \leq (k-1) \left(\frac{n-k+1}{n-\overline{d}-k+1}\right)^k \\
  &=&    (k-1) \left(1+\frac{\overline{d}}{n-\overline{d}-k+1}\right)^k
  \leq (k-1) \exp\left(\frac{\overline{d}k}{n-\overline{d}-k+1}\right),
\end{eqnarray*}
which yields,
\[ \overline{d} \geq \frac{(n-k+1)\log (n/(k-1))}{k + \log (n/(k-1))}.\]
In this paper, we will mainly be interested in $k \in [n^{1/10},
  n^{9/10}]$, in which case we obtain
\[  |E(G)| = n \overline{d} = \Omega\left(\frac{n^2}{k} \log n \right). \]
For the problem under consideration, this immediately gives the
necessary condition
\begin{equation}
\label{eq:kst-edges}
 \sum_{i=1}^r n_i^2 = \Omega\left(\frac{n^2}{k} \log n \right). 
\end{equation}
It will be convenient to normalize $n_i$ and define $\alpha_i =
\frac{n_i}{n/k}$. With this notation, the inequality above can be 
restated as follows.
\begin{equation}
\label{eq:kst}
 \sum_{i=1}^r \alpha_i^2 = \Omega(k \log n).
\end{equation}

\paragraph*{\bf The Hansel bound \\} 

The same question can also be asked in the context of general
graphs. In that case, we have the following classical theorem.

\begin{theorem}[Hansel~\cite{Hansel}, Katona and Szemer\'{e}di~\cite{KSz}, Krichevskii~\cite{Krichevskii}]
\label{thm:hansel}
Suppose it is possible to place one copy each of $K_{n_i,n_i}$,
$i=1,2,\ldots,r$, in a vertex set of size $n$ such that the
resulting graph has no independent set of size $k$. Then,
\[ \sum_{i=1}^r n_i \geq n \log \left( \frac{n}{k-1} \right). \]
\end{theorem}

\noindent Although this result pertains to general graphs and is not
directly applicable to the bipartite graph setting, it can be used
(details omitted as we will use this bound only to motivate our
results, not to derive them) to derive a necessary condition for
bipartite graphs as well. In particular, normalizing $n_i$ by setting
$n_i = \alpha_i\frac{n}{k}$ as before, one can obtain the necessary
condition
\begin{equation}
\label{eq:hansel}
 \sum_{i=1}^r \alpha_i = \Omega(k \log n).
\end{equation}

Note that neither of the two bounds above strictly dominates the
other: if all $\alpha_i$ are small (say $\ll 1$), then the first
condition derived from the K\H{o}v\'{a}ri, S\'{o}s and Tur\'{a}n bound
is stronger, wheras if all $\alpha_i$ are large ($\gg 1$), then the
condition derived from the Hansel bound is stronger.

In our applications, we will meet situations where the $\alpha_i$'s
will not be confined to one or the other regime. To get optimal
results, one must, therefore, devise a condition appropriate for the
entire range of values for the $\alpha_i$'s. Towards this goal, we
start by trying to guess the form of this general inequality by asking
a dual question: what is a sufficient condition on $n_i$'s
(equivalently $\alpha_i$'s) for $(n_1,n_2,\ldots,n_r)$ to be
$(n,k)$-strong? We derive the following (proof omitted).

\begin{theorem}[Sufficient condition]
\label{thm:sufficient.symmetric}
Suppose $k \in [n^{1/10}, n^{9/10}]$, and let $\alpha_i \in
[n^{-1/100}, n^{1/100}]$, $i=1,2\ldots,r$. Then, there is a constant $A > 0$
such that if
\[ \sum_{i: \alpha_i \leq 1} \alpha_i^2 + \sum_{i: \alpha_i > 1} \alpha_i \geq A k \log n,\]
then $(n_1,n_2,\ldots,n_r)$ is $(n,k)$-strong, where $n_i = \alpha_i
(n/k)$.
\end{theorem}

\noindent We might ask if this sufficient condition is also
necessary. The K\H{o}v\'{a}ri, S\'{o}s and Tur\'{a}n bound (Inequality
~\ref{eq:kst}) explains the first term in the LHS of the above
sufficient condition, and the Hansel bound
(Inequality~\ref{eq:hansel}) explains the second term. We thus have
explanations for both the terms using two classical theorems of graph
theory. However, neither of them implies in full generality that the
sufficient condition derived above is necessary. In this work, we show
that this sufficient condition is indeed also necessary upto
constants.

\begin{theorem}[Necessary condition]
\label{thm:necessary.symmetric}
Suppose $k \in [n^{1/10}, n^{9/10}]$, and let $\alpha_i \in
[n^{-1/100}, n^{1/100}]$, $i=1,2\ldots,r$. Then, there is a constant
$B > 0$ such that if $(n_1,n_2,\ldots,n_r)$ is $(n,k)$-strong where $n_i =
\alpha_i (n/k)$, then
\[ \sum_{i: \alpha_i \leq 1} \alpha_i^2 + \sum_{i: \alpha_i > 1} \alpha_i \geq B k \log n.\]
\end{theorem}

%\paragraph*{\bf Remark.} A weaker version of this necessary condition
%was employed by Radhakrishnan and Ta-Shma~\cite{RT} in the proof of
%the optimal lower bound on the size of depth-two
%superconcentrators. They showed that (for some constant $D$), if
%$\sum_{i: \alpha_i > 1} \alpha_i \leq D k \log n$, then $\sum_{i:
%\alpha_i \leq 1} \alpha_i^2 \geq D k$ (with the $\log n$ missing on
%the right hand side).

\noindent Our proof of Theorem~\ref{thm:necessary.symmetric} uses
refinement of the ideas used by Radhakrishnan and Ta-Shma~\cite{RT} 
for proving an optimal lower bound on the size of depth-two superconcentrators.
A tradeoff result for depth-two superconcentrators was shown by Dutta
and Radhakrishnan~\cite{DR}. Their main argument leads one to consider
situations where the small bipartite graphs used to build the bigger
one are not symmetric, instead of being of the form $K_{n_i,n_i}$,
they are of the form $K_{m_i,n_i}$ (with perhaps $m_i \neq n_i$).

\begin{definition}
We say that a sequence $((m_1,n_1),(m_2,n_2),\ldots,(m_r,n_r))$ of
pairs of positive integers is $(n,k)$-strong if there is a bipartite
graph $G=(V,W,E)$ that is a union of graphs $G_i=(V_i,W_i, E_i=V_i
\times W_i)$, $i=1,2,\ldots,r$, such that
\begin{itemize}
\item $|V|, |W|=n$;
\item $|V_i|=m_i$ and $|W_i|=n_i$.
\item $G$ has no bipartite independent set of size $k \times k$.
\end{itemize}
\end{definition}

\noindent We refine our lower bound argument for the symmetric case
and provide necessary condition for this asymmetric setting as
well. 

\begin{theorem}[Necessary condition: asymmetric case]
\label{thm:necessary.asymmetric}
Suppose $\alpha_i, \beta_i \in [n^{-1/100}, n^{1/100}]$,
$i=1,2\ldots,r$, and $k \in [n^{1/10}, n^{9/10}]$. Then, there is a
constant $C > 0$ such that if the sequence
$((m_1,n_1),(m_2,n_2),\ldots,(m_r,n_r))$ is $(n,k)$-strong where $m_i
= \alpha_i (n/k)$ and $n_i = \beta_i (n/k)$, then
%\[ \sum_{i: \min\{\alpha_i, \beta_i\} \leq 1} \alpha_i \beta_i + \sum_{i: \min\{\alpha_i,\beta_i\} > 1} (\alpha_i + \beta_i) \ent(p_i) \geq D k \log n, \]
\[ \sum_{i \in X} \alpha_i \beta_i + \sum_{i \in \{1,2,\ldots,r\} \setminus X} (\alpha_i + \beta_i) \ent(p_i) \geq C k \log n \]
for every $X \subseteq \{1,2,\ldots,r\}$, where $p_i =
\frac{\alpha_i}{\alpha_i + \beta_i}$ and $\ent(p_i) = -p_i \log(p_i) -
(1-p_i) \log(1-p_i)$.
\end{theorem}

\noindent Our proofs of the above mentioned necessary conditions uses
uses techniques developed for analyzing depth-two superconcentrators~\cite{RT,DR}
which are important combinatorial objects useful in both algorithms and
complexity (formal definition in Section~\ref{sec:superconcentrator}). 
As applications of our lower bounds, we provide modular proofs of two
known lower bounds for depth-two superconcentrators:
the first one is a lower bound on the number of edges in such graphs
(Theorem~\ref{thm:superconcentrator.lb}) first shown in~\cite{RT} which we
reprove here using Theorem~\ref{thm:necessary.symmetric} in
Section~\ref{sec:superconcentrator}; the second
one is a tradeoff result between the number of edges at the two
levels of such graphs (Theorem~\ref{thm:superconcentrator.tradeoff}) 
first shown in~\cite{DR} which we reprove here using 
Theorem~\ref{thm:necessary.asymmetric} in Section~\ref{sec:superconcentrator}.
These results were stated in a preliminary version of this work~\cite{DR-Zarankiewicz}.

\junk{
\begin{theorem}[Sufficient condition: asymmetric case]
\label{thm:sufficient.asymmetric}
Suppose $k \in [n^{1/10}, n^{9/10}]$, and let $\alpha_i, \beta_i \in
[n^{1/100}, n^{-1/100}]$, $i=1,2\ldots,r$.  Then, there is a constant
$C > 0$ such that if
%\[ \sum_{i: \min\{\alpha_i, \beta_i\} \leq 1} \alpha_i \beta_i + \sum_{i: \min\{\alpha_i, \beta_i\} > 1} (\alpha_i + \beta_i) \ent(p_i) \geq C k \log n \]
\[ \sum_{i \in X} \alpha_i \beta_i + \sum_{i \in \{1,2,\ldots,r\} \setminus X} (\alpha_i + \beta_i) \ent(p_i) \geq C k \log n \]
for every $X \subseteq \{1,2,\ldots,r\}$, where $p_i =
\frac{\alpha_i}{\alpha_i + \beta_i}$ and $\ent(p_i) = -p_i \log(p_i) -
(1-p_i) \log(1-p_i)$ is the binary entropy function, then the sequence
$((m_1,n_1),(m_2,n_2),\ldots,(m_r,n_r))$ is $(n,k)$-strong where $m_i
= \alpha_i (n/k)$, $n_i = \beta_i (n/k)$.
\end{theorem}
}

\section{Building a Bipartite Graph from Smaller Symmetric Bipartite Graphs}
\label{sec:symmetric}

\paragraph*{\bf Proof of Theorem~\ref{thm:sufficient.symmetric} \\}

Let us consider a probabilistic construction of a bipartite graph
$G=(V,W,E)$ where, given $n_1, n_2, \ldots, n_r$ such that $\alpha_i =
\frac{n_i}{n/k} \in [n^{-1/100}, n^{1/100}]$, we place an
independently drawn random copy $G_i$ of $K_{n_i,n_i}$ between $V$ and
$W$. In other words, $G$ is a union of $G_1,G_2,\ldots,G_r$ where
$G_i = (V_i,W_i, E_i=V_i \times W_i)$ and $V_i$, $W_i$ are uniformly
chosen random $n_i$ element subsets of $V$ and $W$ respectively. Fix a
potential independent set $(S,T)$ of size $k \times k$. Then, as shown
below,
\begin{equation}
\label{eq:basic}
p_i = \Pr[E(G_i) \cap S \times T = \emptyset] \leq 1 - (1 - \exp(-\alpha_i))^2. 
\end{equation}

\noindent Thus, since the graphs $G_i$'s are chosen independently, the
probability that $(S,T)$ is independent in $G$ is
\[ p = \Pr[(S,T) \mbox{ is independent in } G] \leq \prod_{i=1}^r (1- (1 - \exp(-\alpha_i))^2).\]
By the union bound, if $p {n \choose k}^2 < 1$ then there is bipartite
graph built by putting together one copy each of $K_{n_i,n_i}$ that
avoids all independent sets of size $k \times k$. The interesting
aspect of this calculation is in the form of the expression $p_i=1- (1 -
\exp(-\alpha_i))^2$.  We will show below that
\begin{equation}
\label{eq:pi}
p_i \leq \left\{ \begin{array}{l l}
                        \exp\left(-\frac{\alpha_i^2}{3}\right) & \mbox{if $\alpha_i \leq 1$} \\
                        \exp(- (1-\ln 2) \alpha_i) & \mbox{if $\alpha_i > 1$}
                       \end{array} \right. 
\end{equation}

\noindent This immediately gives us our first result, the sufficient
condition stated as Theorem~\ref{thm:sufficient.symmetric} in the
introduction.

\paragraph{Proof of (\ref{eq:basic}):} Recall that $G_i=(V_i, W_i, V_i \times W_i)$, 
where $V_i$ and $W_i$ are uniformly chosen random subsets of $V$ and
$W$ of size $n_i$ each.
\begin{eqnarray*} 
\Pr[E(G_i) \cap S \times T = \emptyset] &=& \Pr[V_i \cap S = \emptyset \vee W_i \cap T = \emptyset] \\
            &=& 1 - (1-\Pr[V_i \cap S = \emptyset])(1-\Pr[W_i \cap T = \emptyset]).
\end{eqnarray*}
Then, (\ref{eq:basic}) follows from this because
\begin{eqnarray*}
\Pr[V_i \cap S = \emptyset],  \Pr[W_i \cap T = \emptyset] 
                            & = & {n-k \choose n_i}{n \choose n_i}^{-1} \\
                            & \leq & \left(\frac{n-k}{n}\right)^{n_i} \\
                            & \leq & \exp\left(-\frac{kn_i}{n}\right) \\
                            & = & \exp(-\alpha_i).
\end{eqnarray*}

\paragraph{Proof of (\ref{eq:pi}):} We have
\[ p_i = 1 - (1 - \exp(-\alpha_i))^2 = \exp(-\alpha_i)(2 - \exp(-\alpha_i)).\]
If $\alpha_i \leq 1$, then we have
\[ 2 - \exp(-\alpha_i) \leq 1 + \alpha_i -\frac{\alpha_i^2}{2} + \frac{\alpha_i^3}{6}
                  \leq \exp\left(\alpha_i - \frac{\alpha_i^2}{3}\right).\]
Thus,
\[ p_i = \exp(-\alpha_i)(2 - \exp(-\alpha_i)) \leq \exp\left(-\frac{\alpha_i^2}{3}\right).\]
If $\alpha_i > 1$, then we have
\begin{eqnarray*} 
p_i &=& \exp(-\alpha_i)(2 - \exp(-\alpha_i)) 
\leq 2 \exp(-\alpha_i) \\
&=& \exp(-\alpha_i + \ln 2) 
\leq \exp(- (1 -\ln 2)\alpha_i).
\end{eqnarray*}

\paragraph*{\bf Proof of Theorem~\ref{thm:necessary.symmetric} \\}

Let $k \in [n^{\frac{1}{10}},n^{\frac{9}{10}}]$ and $\alpha_i \in
[n^{-1/100},n^{1/100}]$, $i=1,2,\ldots,r$. Suppose we are given a
bipartite graph $G=(V,W,E)$ which is a union of complete bipartite
graphs $G_1,G_2,\ldots,G_r$ and has no bipartite independent set of
size $k \times k$, where $G_i = (V_i,W_i,E_i=V_i \times W_i)$ with
$|V_i| = |W_i| = n_i = \alpha_i (n/k)$. We want to show that for some
constant $B > 0$,
\[ \sum_{i: \alpha_i \leq 1} \alpha_i^2 + \sum_{i: \alpha_i > 1} \alpha_i \geq B k \log n.\]

\noindent We will present the argument for the case when $k=\sqrt{n}$;
the proof for other $k$ is similar, and focussing on this $k$ will
keep the notation and the constants simple. We will show that if the
second term in the LHS of the above inequality is small, say,
\[ \secondterm = \sum_{i: \alpha_i > 1} \alpha_i \leq \frac{1}{100} k \log n,
\]
then the first term must be large, i.e.,
\[ \firstterm = \sum_{i: \alpha_i \leq 1} \alpha_i^2 \geq \frac{1}{100} k \log n.\]

Assume $\secondterm \leq \frac{1}{100} k \log n$. Let us call a $G_i$
for which $\alpha_i > 1$ as {\em large} and a $G_i$ for which
$\alpha_i \leq 1$ as {\em small}. We start as in \cite{RT} by deleting
one of the sides of each large $G_i$ independently and uniformly at
random from the vertex set of $G$.  For a vertex $v \in V$, let $d_v$
be number of large $G_i$'s such that $v \in V_i$. The probability that
$v$ survives at the end of the random deletion is precisely
$2^{-d_v}$. Now,
\[ \sum_{v} d_v = \sum_{i: \alpha_i > 1} n_i  \leq \frac{1}{100} n \log n,\]
where the inequality follows from our assumption that $\secondterm
\leq \frac{1}{100} k \log n$. That is, the average value of $d_v$ is
$\frac{1}{100} \log n$, and by Markov's inequality, at least half of
the vertices have their $d_v$'s at most $d=\frac{1}{50} \log n$. We
focus on a set $V'$ of $n/2$ such vertices, and if they survive the
first deletion, we delete them again with probability
$1-2^{-(d-d_v)}$, so that every one of these $n/2$ vertices in $V'$
survives with probability exactly $2^{-d}= n^{-1/50}$.  Let $X$ be the
vertices of $V'$ that survive. Similarly, we define $W' \subseteq W$,
and let $Y \subseteq W'$ be the vertices that survive.

\begin{claim}
With probability $1-o(1)$, $|X|, |Y| \geq \frac{n}{4}2^{-d}$.
\end{claim}

The claim can be proved as follows. For $v \in V'$, let $I_v$ be the
indicator variable for the event that $v$ survives. Then, $\Pr[I_v=1]=
2^{-d} = n^{-1/50}$ for all $v \in V'$.  Furthermore, $I_v$ and $I_{v'}$ are
dependent precisely if there is a common large $G_i$ such that both
$v, v' \in V_i$. Thus, any one $I_v$ is dependent on at most $\Delta =
d_v \times \max\{n_i: \alpha_i > 1\} \leq (1/50) (\log n) n^{1/100}
(n/k) = (1/50) n^{51/100} \log n$ such events (recall $k =
\sqrt{n}$). We thus have (see Alon-Spencer~\cite{AS})
\begin{eqnarray*}
 \E[|X|] & = & \sum_{v \in V'} I_v =\frac{n}{2}2^{-d} = \frac{1}{2}n^{49/50}; \\ 
 \var[|X|] &\leq & E[|X|] \Delta.
\end{eqnarray*}
By Chebyshev's inequality, the probability that $|X|$ is 
less than $\frac{\E[|X|]}{2}$ is at most 
\[ \frac{4\var[X]}{\E[|X|]^2} \leq \frac{4 \Delta}{\E[|X|]} = o(1). \]
A similar calculation can be done for $|Y|$.  (End of Claim.)

The crucial consequence of our random deletion process is that no
large $G_i$ has any edge between $X$ and $Y$.  Since $G$ does not have
any independent set $(S,T)$ of size $k \times k$, the small $G_i$'s
must provide the necessary edges to avoid such independent sets
between $X$ and $Y$. Consider an edge $(v,w)$ of a small $G_i$. The
probability that this edge survives in $X \times Y$ is precisely the
probability of the event $I_v \wedge I_w$. Note that the two events
$I_v$ and $I_w$ are either independent (when $v$ and $w$ do not belong
to a common large $G_i$), or they are mutually exclusive. Thus, the
expected number of edges supplied between $X$ and $Y$ by small $G_i$'s
is at most
\[ \sum_{i: \alpha_i \leq 1} \alpha_i^2 (n/k)^2 2^{-2d} = \firstterm \times (n/k)^2 2^{-2d}, \]
and by Markov's inequality, with probability $1/2$ it is at most twice
its expectation. Using the Claim above we conclude that the following
three events happen simultaneously: (a) $|X| \geq \frac{n}{4} 2^{-d}$,
(b) $|Y| \geq \frac{n}{4} 2^{-d}$, (c) the number of edges conecting
$X$ and $Y$ is at most $2 \times \firstterm \times (n/k)^2
2^{-2d}$. Using (\ref{eq:kst-edges}), this number of edges must be at
least $\frac{1}{3} \frac{(n2^{-d})^2}{16k} (\frac{49}{50}\log n - 2)$.
(Note that $\frac{1}{3}$ suffices as the constant in
(\ref{eq:kst-edges}) for the case $|X|, |Y| \geq \frac{n^{49/50}}{4}$
and $k = \sqrt{n}$.)  Comparing the upper and lower bounds on the
number of edges thus established, we obtain the required inequality
\[ \firstterm \geq \frac{1}{100} k \log n. \]

\section{Building a Bipartite Graph from Smaller Asymmetric Bipartite Graphs}
\label{sec:asymmetric}

\paragraph*{\bf Proof of Theorem~\ref{thm:necessary.asymmetric} \\}

Let $k \in [n^{\frac{1}{10}},n^{\frac{9}{10}}]$ and $\alpha_i, \beta_i
\in [n^{-\frac{1}{100}},n^{\frac{1}{100}}]$, $i=1,2,\ldots,r$. Suppose
we are given a bipartite graph $G=(V,W,E)$ which is a union of
complete bipartite graphs $G_1,G_2,\ldots,G_r$ and has no bipartite
independent set of size $k \times k$, where $G_i=(V_i,W_i,E_i=V_i
\times W_i)$ with $|V_i| = m_i = \alpha_i (n/k)$ and $|W_i| = n_i =
\beta_i (n/k)$. As stated in Theorem~\ref{thm:necessary.asymmetric},
we let $p_i = \frac{\alpha_i}{\alpha_i + \beta_i}$ and $\ent(p_i) =
-p_i \log(p_i) - (1-p_i) \log(1-p_i)$. We wish to show that there is a
constant $C > 0$, such that
\[ \sum_{i \in X} \alpha_i \beta_i + \sum_{i \notin X} (\alpha_i + \beta_i) \ent(p_i) \geq C k \log n \]
for every $X \subseteq \{1,2,\ldots,r\}$.

The proof is similar to but more subtle than the proof of
Thereom~\ref{thm:necessary.symmetric} and again we present the
argument for the case when $k=\sqrt{n}$. Fix a subset $X \subseteq
\{1,2,\ldots,r\}$. Our plan is to assume that the second term in the
LHS of the above inequality is small,
\begin{equation}
 \secondterm = \sum_{i \notin X} (\alpha_i + \beta_i) \ent(p_i) \leq \frac{1}{100} k \log n, \label{eq:assumption-second}
\end{equation}
and from this conclude that the first term must be large,
\begin{equation}
\firstterm = \sum_{i \in X} \alpha_i \beta_i \geq \frac{1}{100} k \log n.
\label{eq:conclusion}
\end{equation}

Assume $\secondterm \leq \frac{1}{100} k \log n$. Graphs $G_i$ where $i
\in X$ will be called {\em marked} and graphs $G_i$ where $i \notin X$
will be called {\em unmarked}. As before, we will delete one of the
sides of each unmarked $G_i$ independently at random from the vertex
set of $G$.  However, since this time there are different number of
vertices on the two sides of $G_i$, we need to be more careful and
choose the deletion probabilities carefully. For every unmarked $G_i$
independently, we delete all the vertices in $W_i$ with probability
$p_i$ and all the vertices in $V_i$ with probability $1-p_i$.

For a vertex $v \in V$, let $S_v$ be the set of $i \notin X$ such that
$v \in V_i$. Define $d_v$ to be the quantity $\sum_{i \in S_v}
\log(1/p_i)$. The probability that $v$ survives the random deletion
process is $2^{-d_v}$. Using the fact that $p_i =
\frac{\alpha_i}{\alpha_i + \beta_i}$ and plugging the expression for
$\ent(p_i)$ in the assumption (\ref{eq:assumption-second}), we obtain
\[ \sum_{i \notin X} \left(\alpha_i \log (1/p_i) + \beta_i \log(1/(1-p_i)) \right) \leq \frac{1}{100} k \log n. \]
Multiplying both sides by $(n/k)$, this implies
\begin{equation}
\label{eq:left}
 \sum_{i \notin X} m_i \log (1/p_i)  \leq \frac{1}{100} n \log n, 
\end{equation}
and
\begin{equation}
\label{eq:right}
 \sum_{i \notin X} n_i \log (1/(1-p_i))  \leq \frac{1}{100} n \log n. 
\end{equation}
Since
\[ \sum_{v \in V} d_v = \sum_{i \notin X} m_i \log (1/p_i), \]
the average value of $d_v$ is at most $\frac{1}{100} \log n$, and by
Markov's inequality, at least $3n/4$ vertices $v \in V$ have their
$d_v$ at most $d =\frac{1}{25} \log n$. Moreover, since $\alpha_i,
\beta_i \in [n^{-1/100}, n^{1/100}]$, we have
\[ p_i \leq \frac{n^{1/100}}{n^{1/100} + n^{-1/100}}
       \leq 1 - \frac{n^{-1/100}}{n^{-1/100}+ n^{1/100}} 
       \leq \exp\left(- \frac{n^{-1/100}}{n^{1/100} + n^{-1/100}}\right),\]
and thus
\[
 \frac{1}{p_i}  \geq  \exp\left(\frac{n^{-1/100}}{n^{1/100} + n^{-1/100}}\right)  \geq  \exp\left(\frac{1}{2} n^{-1/50}\right).
\]
The above implies $\log(1/p_i) \geq \frac{1}{2} n^{-1/50}$, which
combined with (\ref{eq:left}) yields
\[ \sum_{i \notin X} m_i \leq \frac{1}{50} n^{51/50} \log n. \]
Since
\[ \sum_{v \in V} |S_v| = \sum_{i \notin X} m_i, \]
the average value of $|S_v|$ is at most $\frac{1}{50} n^{1/50} \log
n$, and again by Markov's inequality, at least $3n/4$ vertices $v \in
V$ satisfy $|S_v| \leq d' = \frac{4}{50} n^{1/50} \log n$.

We focus on a set $V'$ of $n/2$ vertices $v \in V$ such that $d_v \leq
d$ and $|S_v| \leq d'$. If any vertex $v \in V'$ survives the first
deletion, we delete it further with probability $1-2^{-(d-d_v)}$, so
that the survival probability of each vertex in $V'$ is exactly
$2^{-d}= n^{-1/25}$.  Let $X$ be the set of vertices in $V'$ that
survive. Similarly, we define $W' \subseteq W$, and let $Y$ be the set
of vertices in $W'$ that survive.

\begin{claim}
With probability $1-o(1)$, $|X|, |Y| \geq \frac{n}{4}2^{-d}$.
\end{claim}

The proof of the claim is exactly like the previous time. 
\junk{For $v \in V'$, we let $I_v$ be the indicator variable for the event
that $v$ survives in $X$. Then, $\Pr[I_v=1]= 2^{-d} = n^{-1/25}$ for
all $v \in V'$.  Furthermore, $I_v$ and $I_{v'}$ are dependent
precisely if there is a common unmarked $G_i$ such that both $v, v'
\in V_i$. Thus, any one $I_v$ is dependent on at most $\Delta = |S_v|
\times \max\{m_i: i \notin X\} \leq (4/50) n^{1/50} (\log n) n^{1/100}
(n/k) = (4/50) n^{53/100} \log n$ such events (recall $k =
\sqrt{n}$). Now we have
\begin{eqnarray*}
\E[|X|] & = & \sum_{v \in V'} I_v = \frac{n}{2}2^{-d} = \frac{1}{2}n^{24/25};\\
\var[|X|] & \leq & E[|X|] \Delta.
\end{eqnarray*}
By Chebyshev's inequality, the probability that $|X|$ is 
less than $\frac{\E[|X|]}{2}$ is at most 
\[ \frac{4\var[X]}{\E[|X|]^2} \leq \frac{4 \Delta}{\E[|X|]} = o(1). \]
A similar calculation can be done for $|Y|$.  (End of Claim.)}

Since no unmarked $G_i$ has any edge between $X$ and $Y$, the marked
$G_i$'s must provide enough edges to avoid all
independent sets of size $k \times k$ between $X$ and $Y$. As in the proof of
Theorem~\ref{thm:necessary.symmetric}, we can argue that an edge of a
marked $G_i$ survives in $X \times Y$ with probability at most $2^{-2d}$. Thus the
expected number of edges supplied between $X$ and $Y$ by marked
$G_i$'s is at most
\[ \sum_{i \in X} m_i n_i 2^{-2d} = \sum_{i \in X} \alpha_i \beta_i (n/k)^2 2^{-2d} = \firstterm \times (n/k)^2 2^{-2d}, \]
and by Markov's inequality with probability $1/2$ it is at most twice
its expectation. Thus the event where both $X$ and $Y$ are of size at
least $\frac{n}{4} 2^{-d}$ and the number of edges connecting them is
at most $2 \times \firstterm \times (n/k)^2 2^{-2d}$ occurs with
positive probability. From (\ref{eq:kst-edges}), this number of edges
must be at least $\frac{1}{3} \frac{(n2^{-d})^2}{16k}
(\frac{24}{25}\log n - 2)$. (Note that $\frac{1}{3}$ suffices as the
constant in (\ref{eq:kst-edges}) when $|X|, |Y| \geq
\frac{n^{24/25}}{4}$ and $k = \sqrt{n}$.) Thus we get
\[ \firstterm \geq \frac{1}{100} k \log n. \]

\section{Depth-Two Superconcentrators}
\label{sec:superconcentrator}

\begin{definition}[Depth-two superconcentrators]
Let $G=(V,M,W,E)$ be a graph with three sets of vertices $V$, $M$ and
$W$, where $|V|, |W|=n$, such that all edges in $E$ go from $V$ to $M$
or $M$ to $W$. Such a graph is called a depth-two
$n$-superconcentrator if for every $k \in \{1,2,\ldots,n\}$ and every
pair of subsets $S \subseteq V$ and $T \subseteq W$, each of size $k$,
there are $k$ vertex disjoint paths from $S$ to $T$.
\end{definition}

\noindent We reprove two known lower bounds for depth-two superconcentrators.

\subsection{Size of Depth-Two Superconcentrators}
\label{sec:superconcentrator-size}

\begin{theorem}[Radhakrishnan and Ta-Shma~\cite{RT}]
\label{thm:superconcentrator.lb}
If the graph $G(V,M,W,E)$ is a depth-two $n$-superconcentrator, then
$|E(G)| = \Omega(n \frac{(\log n)^2}{\log\log n})$.
\end{theorem}

\begin{proof}
Assume that the number of edges in a depth-two $n$-superconcentrator
$G$ is at most $(B/100) n \frac{(\log n)^2}{\log\log n}$, where $B$ is
the constant in Theorem~\ref{thm:necessary.symmetric}. By increasing
the number of edges by a factor at most two, we assume that each
vertex in $M$ has the same number of edges coming from $V$ and going
to $W$. For a vertex $v \in M$, let $\deg(v)$ denote the number of
edges that come from $V$ to $v$ (equivalently the number of edges that
go from $v$ to $W$). For $k \in [n^{1/4}, n^{3/4}]$, define
\begin{eqnarray*}
\High(k)   &=& \{v \in M: \deg(v) \geq \frac{n}{k} (\log n)^2 \}; \\
\Medium(k) &=& \{v \in M: \frac{n}{k} (\log n)^{-2} \leq \deg(v) 
                          < \frac{n}{k} (\log n)^2 \}; \\
\Low(k)    &=& \{v \in M: \deg(v) < \frac{n}{k} (\log n)^{-2} \}.
\end{eqnarray*}

\begin{claim}
For each $k \in [n^{1/4}, n^{3/4}]$, the number of edges incident on
$\Medium(k)$ is at least $\frac{B}{2} n \log n$.
\end{claim}
Fix a $k \in [n^{1/4},n^{3/4}]$. First observe that $|\High(k)| < k$,
for otherwise, the number of edges in $G$ would already exceed $n
(\log n)^2$, contradicting our assumption.  Thus, every pair of
subsets $S \subseteq V$ and $T \subseteq W$ of size $k$ each has a
common neighbour in $\Medium(k) \cup \Low(k)$. We are now in a
position to move to the setting of
Theorem~\ref{thm:necessary.symmetric}. For each vertex $v \in
\Medium(k) \cup \Low(k)$, consider the complete bipartite graph
between its in-neighbours in $V$ and out-neighbours in $W$.  The
analysis above implies that the union of these graphs is a bipartite
graph between $V$ and $W$ that has no independent set of size $k
\times k$. For $v \in \Medium(k) \cup \Low(k)$, let $\alpha_v =
\frac{\deg(v)}{n/k}$. Using Theorem~\ref{thm:necessary.symmetric}, it
follows that
\begin{equation}
\sum_{v \in \Medium(k) \cup \Low(k): \alpha_v \leq 1} \alpha_v^2 + \sum_{v \in \Medium(k) \cup \Low(k): \alpha_v > 1} \alpha_v \geq B k \log n.
\end{equation}
For $\alpha_v \leq 1$, $\alpha_v^2 \leq \alpha_v$ and thus we can
replace $\alpha_v^2$ by $\alpha_v$ when $(\log n)^{-2} \leq
\alpha_v \leq 1$ and conclude
\begin{equation}
\label{eq:fixedk}
\sum_{v \in \Low(k)} \alpha_v^2 + \sum_{v \in \Medium(k)} \alpha_v \geq B k \log n.
\end{equation}
One of the two terms in the LHS is at least half the RHS. If
it is the first term then noting that $\alpha_v < (\log n)^{-2}$
for all $v \in \Low(k)$, we obtain
\[ \sum_{v \in \Low} \deg(v) = \frac{n}{k} \sum_{v \in \Low} \alpha_v 
\geq \frac{n}{k} (\log n)^2 \sum_{v \in \Low} \alpha_v^2 \geq \frac{B}{2} n(\log n)^3. \] 
Since the left hand side is precisely the number of
edges entering $\Low(k)$, this contradicts our assumption that $G$ has
few edges. So, it must be that the second term in the LHS of
(\ref{eq:fixedk}) is at least $\frac{B}{2} k \log n$. Then, the number
of edges incident on $\Medium(k)$ is
\[ \sum_{v \in \Medium} \deg(v) = \frac{n}{k} \sum_{v \in \Medium} \alpha_v \geq \frac{B}{2} n \log n. \]
This completes the proof of the claim.

Now, consider values of $k$ of the form $n^{1/4} (\log n)^{4i}$ in the
range $[n^{1/4}, n^{3/4}]$. Note that there are at least
$(\frac{1}{10})\log n / \log \log n$ such values of $k$ and the sets
$\Medium(k)$ for these values of $k$ are disjoint. By the claim above,
each such $\Medium(k)$ has at least $\frac{B}{2} n \log n$ edges
incident on it, that is $G$ has a total of at least $\frac{B}{20} n
\frac{(\log n)^2}{\log \log n}$ edges, again contradicting our
assumption.
\end{proof}

\subsection{Tradeoff in Depth-Two Superconcentrators}
\label{sec:superconcentrator-tradeoff}

\begin{theorem}[Dutta and Radhakrishnan~\cite{DR}]
\label{thm:superconcentrator.tradeoff}
If the graph $G=(V,M,W,E)$ is a depth-two $n$-superconcentrator with
average degree of nodes in $V$ and $W$ being $a$ and $b$ respectively
and $a \leq b$, then
\[ a \log\left(\frac{a + b}{a}\right) \log b = \Omega(\log^2 n). \]
\end{theorem}

\begin{proof}
We may assume that $b > \log n$, otherwise the total number of edges
in $G$ is at most $2n\log n$ which contradicts
Theorem~\ref{thm:superconcentrator.lb} proved earlier. We may also
assume that $b < n^{\frac{1}{10}}$, otherwise the theorem can be
easily seen to be true.  For a vertex $v \in M$, let $\deg_V(v)$
denote the number of edges that come from $V$ to $v$ and $\deg_W(v)$
denote the number of edges that go from $v$ to $W$. We will assume
that the ratio $\frac{\deg_V(v)}{\deg_W(v)}$ is equal to $\frac{a}{b}$
for each vertex $v \in M$. This is without loss of generality as we
can make the ratio $\frac{\deg_V(v)}{\deg_W(v)}$ equal to
$\frac{a}{b}$ by increasing the number of edges from $V$ to $v$ (if
$\frac{\deg_V(v)}{\deg_W(v)}$ is smaller) or increasing the number of
edges from $v$ to $W$ (if the ratio is larger), and this process does
not increase the number of edges between $V$ and $M$ or between $M$
and $W$ more than by a factor two. (We ignore the rounding issues as
they are not important.)

For $k \in [n^{1/4}, n^{3/4}]$, define
\begin{eqnarray*}
\High(k)   &=& \{v \in M: \deg_W(v) \geq \frac{n}{k} b^2 \}; \\
\Medium(k) &=& \{v \in M: \frac{n}{k} b^{-2} \leq \deg_W(v) 
                          < \frac{n}{k} b^2 \}; \\
\Low(k)    &=& \{v \in M: \deg_W(v) < \frac{n}{k} b^{-2} \}.
\end{eqnarray*}

We consider values of $k$ of the form $n^{1/4} b^{4i}$ in the range
$[n^{1/4},n^{3/4}]$.  There are at least $L = \frac{\log n}{10 \log
  b}$ such values of $k$ and the sets $\Medium(k)$ for these values of
$k$ are disjoint. Thus out of these
values of $k$, we can find one, say $k_0$, such that the number
of edges from $V$ to $\Medium(k_0)$ is at most $\frac{an}{L}$. 

We observe that $|\High(k_0)| < k_0$, otherwise the number of edges
between $M$ and $W$ would be at least $b^2 n > bn$ which is a
contradition. Thus every pair of subsets $S \subseteq V$ and $T
\subseteq W$ of size $k_0$ each has a common neighbour in
$\Medium(k_0) \cup \Low(k_0)$. For each vertex $v \in \Medium(k_0)
\cup \Low(k_0)$, consider the complete bipartite graph between the
in-neighbours and out-neighbours of $v$.  The union of these graphs is
a bipartite graph between $V$ and $W$ that has no independent set of
size $k \times k$. For $v \in \Medium(k_0) \cup \Low(k_0)$, let
$\alpha_v = \frac{\deg_V(v)}{n/k_0}$ and $\beta_v =
\frac{\deg_W(v)}{n/k_0}$. It follows from
Theorem~\ref{thm:necessary.asymmetric} that
\begin{equation}
\label{eq:averagek}
\sum_{v\in \Low(k_0)} \alpha_v \beta_v + \sum_{v \in \Medium(k_0)} (\alpha_v + \beta_v) \ent(\frac{\alpha_v}{\alpha_v + \beta_v}) \geq C k_0 \log n,
\end{equation}
where $C$ is the constant from Theorem~\ref{thm:necessary.asymmetric}.
One of the two terms in the LHS is at least half the RHS.
If it is the first term, noting that $\beta_v < b^{-2}$
for all $v \in \Low(k_0)$, we obtain
\[ \sum_{v \in \Low(k_0)} \deg_V(v) =  \frac{n}{k_0} \sum_{v \in \Low(k_0)} \alpha_v 
\geq \frac{n}{k_0} b^2 \sum_{v \in \Low(k_0)} \alpha_v \beta_v 
> \frac{C}{2} n (\log n)^3, \]
as $b > \log n$. Since the left hand side is precisely the number of edges entering
$\Low(k_0)$, we get $a > \frac{C}{2} (\log n)^3$ which proves the
theorem.  If the second term in the LHS of (\ref{eq:averagek}) is at
least $\frac{C}{2} k_0 \log n$, we get
\[ \sum_{v \in \Medium(k_0)} (\alpha_v + \beta_v) \ent(\frac{\alpha_v}{\alpha_v + \beta_v}) \geq \frac{C}{2} k_0 \log n. \]
Simplifying we get
\[ \sum_{v \in \Medium(k_0)} \left(\alpha_v \log\left(\frac{\alpha_v + \beta_v}{\alpha_v}\right) 
   + \beta_v \log \left(\frac{\alpha_v + \beta_v}{\beta_v}\right) \right) \geq \frac{C}{2} k_0 \log n. \]
We know that $\left(\frac{\alpha_v + \beta_v}{\beta_v}\right)^{\beta_v} = \left(1 + \frac{\alpha_v}{\beta_v}\right)^{\beta_v} \leq \exp(\alpha_v)$,
which means we have that $\beta_v \log \left(\frac{\alpha_v + \beta_v}{\beta_v}\right) \leq \frac{\alpha_v}{\ln 2}$.
Noting $\frac{\alpha_v + \beta_v}{\alpha_v} = \frac{a + b}{a}$, we have
\[ \sum_{v \in \Medium(k_0)} \alpha_v \left( \log\left(\frac{a + b}{a}\right) + \frac{1}{\ln 2} \right) \geq \frac{C}{2} k_0 \log n. \]
Since $a \leq b$, $\frac{a + b}{a} \geq 2$ and we conclude
\[ \sum_{v \in \Medium(k_0)} \alpha_v \log\left(\frac{a + b}{a} \right) = \Omega(k_0 \log n). \]
The number of edges from $V$ to $\Medium(k_0)$ is precisely 
$\sum_{v \in \Medium(k_0)} \deg_V(v) = (n/k_0) \sum_{v \in \Medium(k_0)} \alpha_v$,  
which is at most $\frac{an}{L}$. Thus
\[ \frac{an}{L} \log \left(\frac{a + b}{a}\right) \geq (n/k_0) \sum_{v \in \Medium(k_0)} \alpha_v \log \left(\frac{a + b}{a}\right) = \Omega(n \log n). \]
Plugging $L = \frac{\log n}{10 \log b}$, we get
\[ a \log\left(\frac{a + b}{a}\right) \log b = \Omega(\log^2 n). \]
\end{proof}

\bibliographystyle{alpha}
\bibliography{zarankiewicz}

\end{document}